\documentclass[envcountsect]{llnsc}

\usepackage{graphicx}

\def\Underline{\setbox0\hbox\bgroup\let\\\endUnderline}
\def\endUnderline{\vphantom{y}\egroup\smash{\underline{\box0}}\\}
\def\|{\verb|}

\usepackage{latexsym}
\usepackage{amsmath}
\usepackage{cases}
\usepackage{amssymb}
\usepackage{amsfonts}
\usepackage{comment}
\usepackage{ascmac}
\usepackage{algorithm}
\usepackage[noend]{algpseudocode}
\usepackage{colortbl}
\usepackage{url}
\usepackage{cite}

\newtheorem{envT}{Theorem}[section]
\newtheorem{envL}[envT]{Lemma}
\newtheorem{envD}[envT]{Definition}
\newtheorem{envC}[envT]{Corollary}
\newtheorem{envP}[envT]{Proposition}

\newcommand{\p}{\partial}

\newcommand{\ep}{\varepsilon}

\newcommand{\R}{\mathbb{R}}
\makeatletter%
\input{ot1txtt.fd}
\makeatother%

\begin{document}

\title{On R\'{e}nyi Differential Privacy \\in Statistics-Based Synthetic Data Generation}

\author{Takayuki Miura\inst{1} \and
Toshiki Shibahara\inst{1} \and
Masanobu Kii\inst{1} \and 
Atsunori Ichikawa\inst{1} \and
Juko Yamamoto\inst{1} \and 
Koji Chida\inst{2}
}
%
%
\institute{NTT Social Informatics Laboratories, Japan, \email{tkyk.miura@ntt.com} \and Gunma University, Faculty of Informatics, Japan
}


\maketitle

\begin{abstract}
Privacy protection with synthetic data generation often uses differentially private statistics and model parameters to quantitatively express theoretical security.
However, these methods do not take into account privacy protection due to the randomness of data generation.
In this paper, we theoretically evaluate R\'{e}nyi differential privacy of the randomness in data generation of a synthetic data generation method that uses the mean vector and the covariance matrix of an original dataset.
Specifically, for a fixed $\alpha > 1$, we show the condition of $\ep$ such that the synthetic data generation satisfies $(\alpha, \ep)$-R\'{e}nyi differential privacy under a bounded neighboring condition and an unbounded neighboring condition, respectively.
In particular, under the unbounded condition, when the size of the original dataset and synthetic dataset is 10 million, the mechanism satisfies $(4, 0.576)$-R\'{e}nyi differential privacy.
We also show that when we translate it into the traditional $(\ep, \delta)$-differential privacy, the mechanism satisfies $(4.00, 10^{-10})$-differential privacy.

\end{abstract}

\section{Introduction}

Personal data is expected to be utilized in various fields such as finance, healthcare, and medicine, but sharing personal data collected by one organization with another organization requires attention to individual privacy.
Traditional anonymization techniques such as $k$-anonymization~\cite{sweeney2002k} and randomized response~\cite{warner1965randomized} have struggled to find a good trade-off between utility and privacy for high-dimensional data~\cite{aggarwal2005k}.
In contrast, a synthetic data generation technique has emerged as a privacy protection method that preserves data utility even for high-dimensional data such as images and tabular data with multi-attributes~\cite{bond2021deep}.
In synthetic data generation, values, which we call {\bf generative parameters}, are extracted from the original raw dataset, and then synthetic data are generated randomly as shown in Fig.~\ref{pic:generative_intro}(a).
The synthetic data are the same format as the original data and statistically similar to them.
Typical generative parameters are statistics of original data and trained parameters of deep neural networks~\cite{sklar1959fonctions,li2014dpsynthesizer,asghar2020differentially,gambs2021growing,zhang2017priv,zhang2021privsyn,mckenna2022aim,goodfellow2014generative,xu2019modeling,DBLP:journals/corr/KingmaW13,rezende2015variational}.
After the synthetic data are generated, they are shared with other organizations, but the generative parameters are typically discarded without being disclosed.

To guarantee privacy protection theoretically, differential privacy~\cite{dwork2006differential} is used as a standard framework.
By adding randomness in generative parameter calculation, the generative parameters become differentially private~\cite{mckenna2021winning, zhang2017priv,abadi2016deep}.
The post-processing property of differential privacy guarantees that synthetic data generated with differentially private generative parameters also satisfy differential privacy as shown in Fig.~\ref{pic:generative_intro}(b).
Although the synthetic data generated with non-differentially private generative parameters have high utility, those with differentially private parameters are known to have lower utility~\cite{tao2021benchmarking}.

\begin{figure}[t]
\begin{center}
\includegraphics[width=0.8\linewidth]{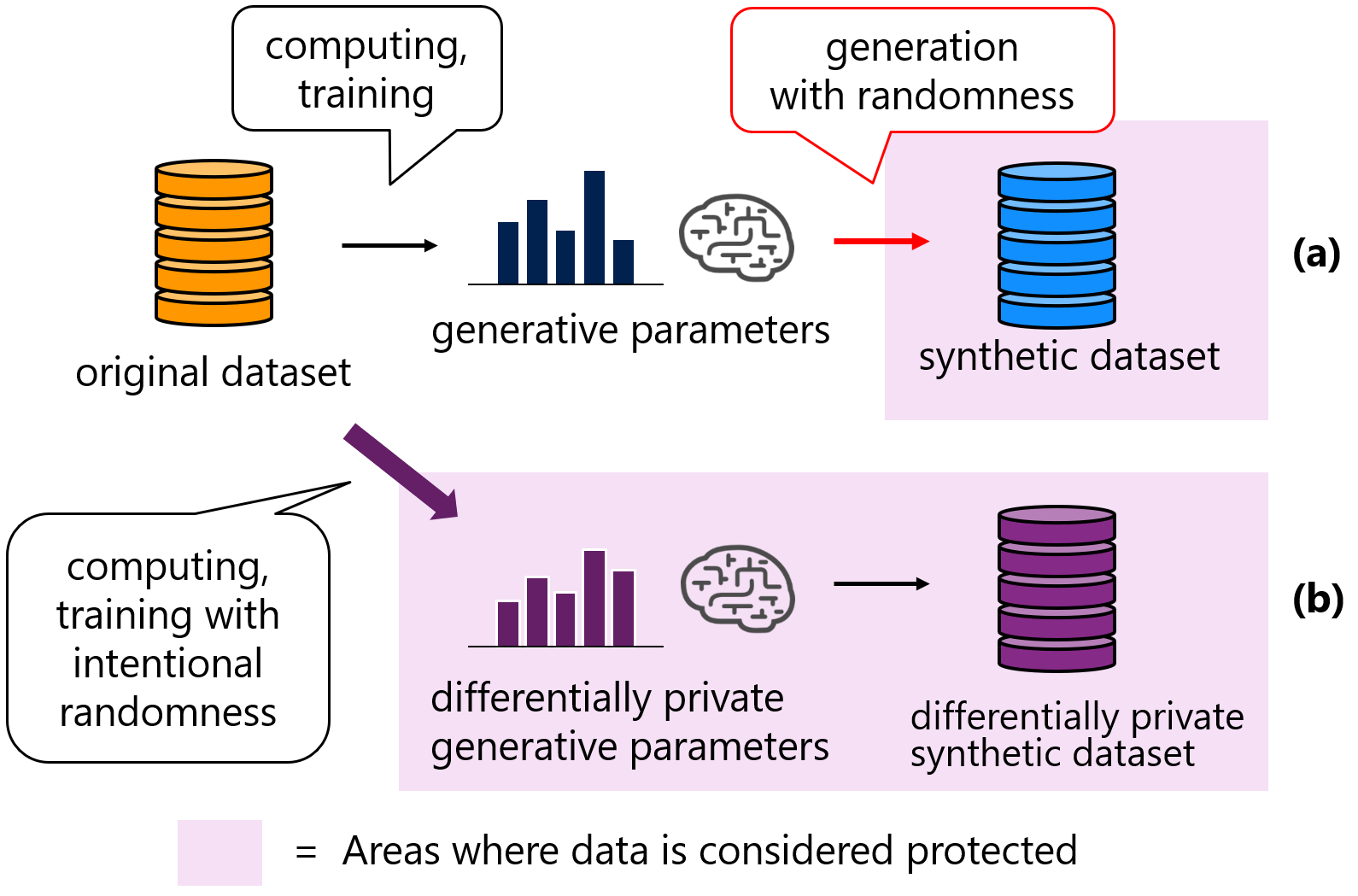}
\caption{(a) Output = Only synthetic data: The generative parameters are discarded after data are generated. We evaluate privacy protection by the randomness in generation.
\\(b) Output = Generative parameters: By computing or training generative parameters with intentional randomness, we obtain differentially private generative parameters that also generate differentially private synthetic data.}
\label{pic:generative_intro}
\end{center}
\end{figure}

We address this problem by evaluating differential privacy of randomness in data generation when using non-differentially private generative parameters.
As mentioned above, in the context of anonymization, the generative parameters are often discarded without disclosing them to the public.
When the output is not generative parameters but only synthetic data, we can consider that it has already been protected by the randomness even if the generative parameters are not protected with differential privacy as shown in Fig.~\ref{pic:generative_intro}(a).
If privacy protection in data generation is quantitatively evaluated, theoretically guaranteed synthetic data can be obtained without degrading the utility.
Moreover, by incorporating this result into traditional methods, we expect to keep the same level of security with smaller additional randomness; that is, we can obtain higher utility synthetic data.

In this paper, we regard a record as a $d$-dimensional vector and focus on a synthetic data generation mechanism with the mean vector and the covariance matrix of the original dataset shown in Fig.~\ref{pic:gen-algo}.
We theoretically evaluate R\'{e}nyi differential privacy~\cite{mironov2017renyi}, which is a relaxed concept of differential privacy, by randomness in generation for the method.
We explicitly derive the condition of $\ep$ such that the synthetic data generation mechanism satisfies $(\alpha, \ep)$-R\'{e}nyi differential privacy for a fixed $\alpha >1$ under the unbounded neighboring condition (Theorem~\ref{thm:main-private}) and the bounded neighboring condition (Corollary~\ref{cor:main-public}).
Furthermore, we conduct a numerical evaluation with reference to the Adult dataset~\cite{Dua:2019} and compute $\ep$ concretely.
We demonstrate that when the size of original dataset is 10 million and the mechanism outputs data the same size as the input dataset, it satisfies $(4, 0.576)$-R\'{e}nyi differential privacy under the unbounded condition and $(4,2.307)$-R\'{e}nyi differential privacy under the bounded condition (Table~\ref{tab:n-table}).
If they are translated into the traditional $(\ep,\delta)$-differential privacy, the mechanism satisfies $(4.00, 10^{-10})$ and $(7.88,10^{-10})$ differential privacy under the unbounded and bounded condition, respectively (Table~\ref{tab:epdelta-table}).
These values are mostly similar to ones used by Apple~\cite{apple} and US Census~\cite{uscensus}.

\section{Preliminaries}

In this section, we introduce basic notations and concepts for later discussion.

\subsection{Notations}
In this paper, we denote the determinant of a square matrix $A \in \R^{d \times d}$ by $|A|:=\det A$.
The transposes of a vector $x \in \R^d$ and a matrix $A \in \R^{d_1 \times d_2}$ are denoted by ${}^t x$ and ${}^t A$.
We assume that datasets are tabular but all discussions can be applied to other datasets such as images since we consider records as vectors.
In a tabular dataset, a record is expressed as a combination of several attribution values.
Each attribution value is a numerical value and normalized into a range $[-1, 1]$.
Thus, a record is regarded as a vector $x \in [-1, 1]^d$, and a dataset with $n$ records is regarded as $D = \{ x_i \}_{i=1, \ldots , n} \in [-1,1]^{d \times n} =: \mathcal{D}$.

\subsection{Differential Privacy}

In this subsection, we introduce $(\ep, \delta)$-differential privacy and $(\alpha, \ep)$-R\'{e}nyi differential privacy.
First, we define neighboring datasets.
\begin{envD}[Neighboring datasets]
Datasets $D, D' \in \mathcal{D}$ are {\bf neighboring datasets} if $D$ and $D'$ are different only in one record.
When datasets have a fixed size $n$, we call the neighboring condition a {\bf bounded condition}~{\rm \cite{kifer2011no}}.
In this case, neighboring means changing the value of exactly one record.
When datasets have no such restriction, we call the neighboring condition an {\bf unbounded condition}~{\rm \cite{kifer2011no}}.
In this case, neighboring means either adding or removing one record.\footnote{This difference is important for the sensitivity of queries. For example, the sensitivity of the mean value query under the bounded condition is twice as large as that under the unbounded condition.}
\end{envD}

$(\ep, \delta )$-differential privacy~\cite{dwork2006differential} is defined as follows.
\begin{envD}[differential privacy~\cite{dwork2006differential}]
A randomized function $\mathcal{M} : \mathcal{D} \to \mathcal{Y}$ satisfies $(\ep, \delta)$-differential privacy ($(\ep, \delta)$-DP) if for any neighboring $D, D' \in \mathcal{D}$ and $S \subset \mathcal{Y}$
\[
\Pr[\mathcal{M}(D) \in S] \le e^{\ep} \Pr[\mathcal{M}(D') \in S] + \delta.
\]
In particular, $\mathcal{M}$ satisfies $\ep$-DP if it satisfies $(\ep, 0)$-DP.
\end{envD}

Next, we define R\'{e}nyi divergence, which is necessary to define R\'{e}nyi differential privacy.

\begin{envD}[R\'{e}nyi Divergence]
Let $P, Q$ be probability distributions on $\R^d$.
For $\alpha > 1$, the {\bf R\'{e}nyi Divergence} of order $\alpha$ is 
\[
D_{\alpha}(P||Q) := \frac{1}{\alpha -1} \log \left( \int_{\R^d} P(x)^{\alpha}Q(x)^{1-\alpha} dx \right).
\]
\end{envD}

\begin{envD}[R\'{e}nyi differential privacy~\cite{mironov2017renyi}]
For $\alpha > 1$ and $\ep > 0$, a randomized function $\mathcal{M}: \mathcal{D} \to \R^d$ satisfies $(\alpha, \ep)$-{\bf R\'{e}nyi differential privacy} ($(\alpha, \ep)$-RDP) if for neighboring datasets $D, D' \in \mathcal{D}$, 
\[
D_{\alpha}(\mathcal{M}(D)||\mathcal{M}(D')) \le \ep .
\]
\end{envD}

The smaller $\ep$ is, the stronger the protection, and the larger $\alpha$ is, the stronger the protection.
To satisfy $(\alpha, \ep)$-RDP for any $\alpha$ is equivalent to $\ep$-DP.

The composition theorem~\cite{dwork2014algorithmic,kairouz2015composition} holds for R\'{e}nyi differential privacy as well as $(\ep, \delta)$-DP.
Furthermore, R\'{e}nyi differential privacy can be translated into $(\ep, \delta )$-DP.

\begin{envP}[Composition of R\'{e}nyi differential privacy~\cite{mironov2017renyi}]\label{prop:composition}
Let $\mathcal{M}_1 : \mathcal{D} \to \R^{d_1}$ be $(\alpha, \ep_1)$-RDP and $\mathcal{M}_2 : \mathcal{D} \times \R^{d_1} \to \R^{d_2}$ $(\alpha, \ep_2)$-RDP.
Then the mechanism $\mathcal{M} : \mathcal{D} \to \R^{d_1} \times \R^{d_2}$ defined as $\mathcal{M}(D) = (\mathcal{M}_1(D), \mathcal{M}_2(D, \mathcal{M}_1(D)))$ satisfies $(\alpha, \ep_1 + \ep_2)$-RDP.
\end{envP}

\begin{envP}[Translation from $(\alpha,\ep)$-RDP to $(\ep, \delta )$-DP~\cite{mironov2017renyi}]\label{prop:rdp-dp}
If $\mathcal{M}$ is an $(\alpha, \ep)$-RDP mechanism, it also satisfies $(\ep + \frac{\log \frac{1}{\delta}}{\alpha -1}, \delta)$-DP for any $0<\delta<1$.
\end{envP}

By the following lemma, the result with the unbounded condition can be reduced to the bounded condition.
\begin{envL}[Weak triangle inequality~\cite{mironov2017renyi}]\label{lem:triangle}
Let $P, Q, R$ be probability distributions on $\R^d$.
For $\alpha > 1$ and $\frac{1}{p} + \frac{1}{q} = 1$, it holds
\[
    D_{\alpha}(P||Q) \le \frac{\alpha - \frac{1}{p}}{\alpha -1} D_{p\alpha}(P||R) + D_{q(\alpha-\frac{1}{p})}(R||Q).
\]
\end{envL}

\subsection{Synthetic Data Generation with Mean Vector and Covariance Matrix}

In this paper, we focus on a simple synthetic data generation with the mean vector and the covariance matrix of the original dataset $\mathcal{M}_G : \mathcal{D} \to [-1, 1]^d $ as shown in Fig.~\ref{pic:gen-algo}.
This method is identical to the Gaussian copula~\cite{sklar1959fonctions} with the assumption that the marginal distributions are all normal distributions.

The mechanism $\mathcal{M}_G$ generates synthetic data as follows.
First, for dataset $D = \{ x_i \}_{i=1, \ldots , n} \in \mathcal{D}$, the mean vector $\mu \in \R^d$ and the covariance matrix $\Sigma \in \R^{d \times d}$ are computed:
\[
\mu := \frac{1}{n} \sum_{i=1}^n x_i, \ \ \Sigma := \frac{1}{n} \sum_{i=1}^n x {}^tx    - \mu {}^t \mu.
\]
Next, a sample is drawn from a multivariate normal distribution $\mathcal{N}(\mu, \Sigma)$, and its values are cut into the range $[-1, 1]^d$.

We denote by $\mathcal{M}_G^n : \mathcal{D} \to [-1,1]^{d \times n}$ the mechanism that simultaneously outputs $n$ records by $\mathcal{M}_G$.
By Proposition~\ref{prop:composition}, we see that if $\mathcal{M}_G$ satisfies $(\alpha, \ep)$-RDP, then $\mathcal{M}_G^n$ also satisfies $(\alpha, n\ep)$-RDP.

\begin{figure}[t]
\begin{center}
\includegraphics[width=0.7\linewidth]{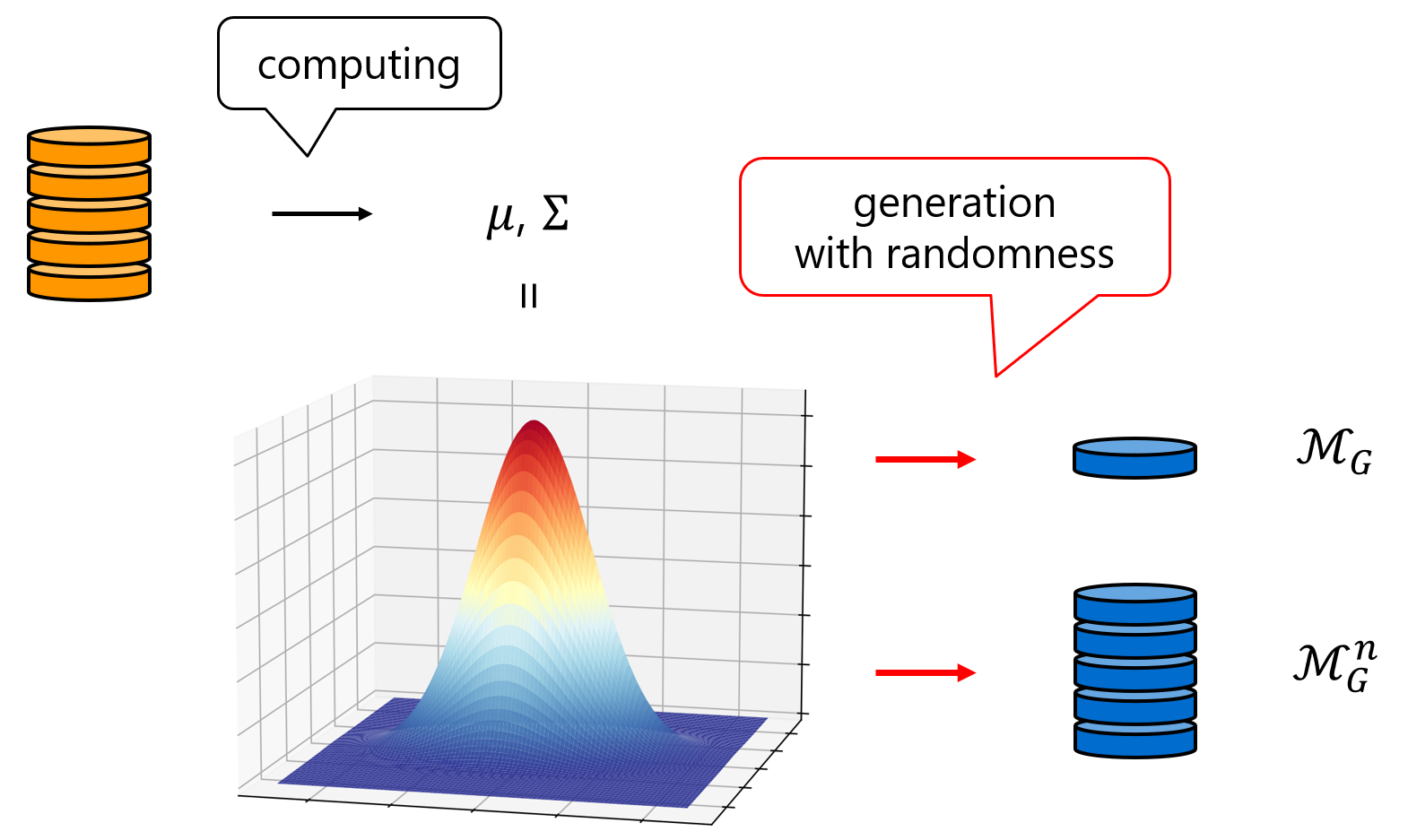}
\caption{Synthetic data generation algorithms $\mathcal{M}_G$ and $\mathcal{M}_G^n$}
\label{pic:gen-algo}
\end{center}
\end{figure}

\subsection{Properties of Symmetric Matrices}

We explain properties of symmetric matrices for the proof of the main theorem.
\begin{envD}[symmetric matrix]
A square matrix $A$ is called {\bf symmetric} if $A = {}^t A$ holds.
\end{envD}
\begin{envD}[positive-definite, semi-positive definite]
For a $d$-dimensional symmetric matrix $A$, the following two conditions are equivalent: \\
{\rm (1)} For all $x \in \R^d\backslash \{0\}$, it holds ${}^t x Ax >0$ $( \ge 0 )$; \\
{\rm (2)} All eigenvalues of $A$ are positive $($non-negative$)$. \\
If $A$ satisfies these conditions, then $A$ is called {\bf positive-definite (positive semi-definite)}.
\end{envD}

The following two lemmas are well-known facts~\cite{harville1998matrix}.
\begin{envL}\label{lem:a}
    Let $A, B$ be positive-definite symmetric matrices.
    If $AB$ is symmetric, then $AB$ is also positive-definite.
\end{envL}
\begin{envL}\label{lem:b}
    Let $A$ be a positive-definite symmetric matrix.
    For an invertible matrix $S$ that is the same size as $A$, ${}^t SAS$ is also positive-definite.
\end{envL}

\begin{envP}\label{prop:matrices}
Let $A, B, C$ be positive-definite symmetric real matrices. If $ABC$ is symmetric, then $ABC$ is also positive-definite.
\end{envP}
\begin{proof}
Set $D := ABC  = CBA$.
Since $C$ is positive-definite, we can obtain the spectral decomposition $C := \sum_{i=1}^d \lambda_i \theta_i {}^t \theta_i$, 
where $\lambda_i > 0$ for all $i = 1, \ldots , d$.
Then we set $S := \sum_{i=1}^d \sqrt{\lambda_i} \theta_i {}^t \theta_i$.
We see that $S$ is symmetric and $C = S^2$ holds.
We have 
\[
S^{-1} D S^{-1} = S^{-1} A S^{-1} SBS = SBS S^{-1} A S^{-1}.
\]
By applying $S^{-1} A S^{-1}$ and $SBS$ to Lemma~\ref{lem:a} and Lemma~\ref{lem:b}, we see that $S^{-1} D S^{-1}$ is positive-definite. Thus, $D$ is also positive-definite.
\end{proof}

\section{Main Theorem}\label{sec:main}

In this paper, we prove the upper bound of $\ep$ such that the mechanism $\mathcal{M}_G$ satisfies $(\alpha, \ep)$-R\'{e}nyi differential privacy for a fixed $\alpha$.
We assume that all datasets have a limitation for the minimum eigenvalue of their covariance matrices.
Specifically, for a fixed $\sigma>0$, we define the set of datasets as 
\[
\mathcal{D}_{\sigma} := \{D \in [-1,1]^{n \times d} \mid z \in S^{d-1}, {}^t z\Sigma_D z \ge \sigma \}.
\]
We also set $\tau := \frac{4d}{\sigma}$.

First, the result under the unbounded condition is the following theorem.
We assume that the number of records in an original dataset is $n$ and that in its neighboring dataset is $n+1$.

\begin{envT}\label{thm:main-private}
Under the unbounded condition, let $\alpha > 1$.
We assume that
\begin{equation}\label{eqn:thm-cond}
\frac{n}{n+1} < \tau, \ \alpha < \min \left\{ n+1, \frac{n^2}{\tau(n+1) -n} \right\}.
\end{equation}
Then, the synthetic data generation mechanism $\mathcal{M}_G$ satisfies $(\alpha, \ep_{\alpha})$-RDP for $\ep_{\alpha}:=\max \{ \ep_{\alpha 1}, \ep_{\alpha 2} \}$.
Here, 
\begin{eqnarray*}
\ep_{\alpha 1} &=& \frac{\alpha}{2}  \cdot \frac{\tau}{(n+1)(n+1-\alpha)} 
+ \frac{\alpha d}{2(\alpha -1)} \log \frac{n}{n+1}  - \frac{d}{2(\alpha -1)} \log \left(1 - \frac{\alpha}{n+1} \right) \\
&&- \frac{1}{2(\alpha -1)} \log \min \left\{ 1, \frac{1 + \alpha \frac{n \tau}{(n+1)(n+1-\alpha)}}{(1 + \frac{ \tau}{n+1})^{\alpha}} \right\}
\end{eqnarray*}
and 
\begin{eqnarray*}
\ep_{\alpha 2} &=& \frac{\alpha}{2} \cdot \frac{\tau}{n(n+\alpha) -\alpha (n+1) \tau}  
 +  \frac{\alpha d}{2(\alpha -1)} \log \frac{n+1}{n} - \frac{d}{2(\alpha -1)} \log \left(1 + \frac{\alpha}{n}\right) \\
&& - \frac{1}{2(\alpha -1)} \log \min \Biggl\{ 1, \frac{1 - \frac{\alpha (n+1) \tau}{(n  + \alpha)n} }{(1 - \frac{\tau}{n})^{\alpha}} \Biggr\} .
\end{eqnarray*}
\end{envT}

Next, under the bounded condition, we obtain the following statement as a corollary of Theorem~\ref{thm:main-private}.
\begin{envC}\label{cor:main-public}
Under the bounded condition, let $\alpha>1$.
We set
\[
c := \min \left\{ n+1, \frac{n^2}{\tau(n+1)-n} \right\}
\]
and assume that
\begin{equation}\label{eqn:cor-cond}
   \alpha < \frac{c^2}{2c-1}. 
\end{equation}
Then, the synthetic data generation mechanism $\mathcal{M}_G$ satisfies $(\alpha, \ep_{\alpha})$-RDP for the following $\ep$:
\begin{equation}\label{eqn:bound-public} 
    \ep_{\alpha} = \inf_{\frac{c-1}{c-\alpha} < p < \frac{c}{\alpha}} \frac{\alpha - \frac{1}{p}}{\alpha -1} \ep\left(p \alpha, n\right) + \ep \left(\frac{p \alpha -1}{p-1}, n+1 \right),
\end{equation}
where $\ep(\alpha, n)$ is the $\ep$ in Theorem~\ref{thm:main-private}.
\end{envC}
\begin{proof}
For any neighboring datasets $D_1, D_2$ under the bounded condition, there exists a dataset $D_3$ such that $D_1$ and $D_3$ are neighboring and $D_2$ and $D_3$ are neighboring under the unbounded condition.
Then, to obtain Equation~(\ref{eqn:bound-public}), we use Lemma~\ref{lem:triangle}.
Here, the weak triangle inequality holds for all $p>1$, and the following condition is necessary:
\[
\max \left\{ p\alpha, \frac{p \alpha -1}{p-1} \right\} < c.
\]
This is equivalent to 
\[
\frac{c-1}{c-\alpha} < p < \frac{c}{\alpha}.
\]
The existence of $p$ is equivalent to Equation~(\ref{eqn:cor-cond}).
\end{proof}

\section{Proof of Theorem~\ref{thm:main-private}}\label{sec:proof}

In this section, we prove Theorem~\ref{thm:main-private}.
The following proposition is essential.
\begin{envP}[Gil et al.~\cite{gil2013renyi}]\label{prop:rdiv-bound}
Let $\alpha>1$ and $\mathcal{N}(\mu_1, \Sigma_1)$, $\mathcal{N}(\mu_2, \Sigma_2)$ be multivariate normal distributions.
If a matrix \[
T_{\alpha} := \alpha\Sigma_1^{-1} + (1-\alpha)\Sigma_2^{-1} 
\] is positive-definite, then it holds
\begin{align*}
&D_{\alpha}(\mathcal{N}(\mu_1, \Sigma_1)|| \mathcal{N}(\mu_2, \Sigma_2))\\
&= \frac{\alpha}{2}{}^t(\mu_1 - \mu_2) \Sigma_{\alpha}^{-1}(\mu_1 - \mu_2)-\frac{1}{2(\alpha -1)}\log \frac{|\Sigma_{\alpha}|}{|{\Sigma_1}|^{1-\alpha}|{\Sigma_2}|^{\alpha}}, 
\end{align*}
where $\Sigma_{\alpha} := (1-\alpha)\Sigma_1 + \alpha \Sigma_2$.
\end{envP}

For neighboring datasets $D_1, D_2 \in \mathcal{D}_{\sigma}$, we set the mean vectors as $\mu_1, \mu_2$ and the covariance matrices as $\Sigma_1, \Sigma_2$.
If $D_{\alpha}(\mathcal{N}(\mu_1, \Sigma_1)|| \mathcal{N}(\mu_2, \Sigma_2)) \le \ep$, the mechanism $\mathcal{M}_G$ satisfies $(\alpha, \ep)$-RDP.
Here we set
\[
L_1 := {}^t(\mu_1 - \mu_2) \Sigma_{\alpha}^{-1}(\mu_1 - \mu_2), \ \ L_2 := \frac{|\Sigma_{\alpha}|}{|{\Sigma_1}|^{1-\alpha}|{\Sigma_2}|^{\alpha}}.
\]
Then we see
\[
D_{\alpha}(\mathcal{N}(\mu_1, \Sigma_1)|| \mathcal{N}(\mu_2, \Sigma_2)) = \frac{\alpha}{2} L_1 - \frac{1}{2(\alpha-1)} \log L_2.
\]
Thus, an upper bound $\ep$ is described by the maximum of $L_1$ and the minimum of $L_2$.
The outline of proof is as follows.
First, by using the different record, we represent the difference between mean vectors and the difference between covariance matrices (Lemma~\ref{lem:private-diff}).
Next, we determine the positive-definiteness of $T_{\alpha}$ (Lemma~\ref{lem:private-pos}).
Finally, we compute the upper bound of $L_1$ (Lemma~\ref{lem:private-l1-max}) and the lower bound of $L_2$ (Lemma~\ref{lem:private-l2-min}).

Set $\# D_1 = n$ and $\#D_2 = n+s$, where $s=1$ when we "add" a record and $s=-1$ when we "remove" a record.
The common records are denoted by $x_1, \ldots , x_n \in [-1, 1]^d$ and the different record by $x \in [-1,1]^d$.
We set each mean vector as $\mu_1, \mu_2$ and covariance matrix as $\Sigma_1, \Sigma_2$.
We also denote by $\sigma_{min}$ the minimum eigenvalue of $\Sigma_1$.
Note that $\sigma_{min} \ge \sigma$ by the assumption.

\begin{envL}[Representations of difference]\label{lem:private-diff}
The following equations hold:
\[
    \mu_{d} := \mu_2 - \mu_1 = \frac{s}{n+s}x - \frac{s}{n(n+s)}\sum_{i=1}^n x_i, 
    \]
    \[
    X:= \Sigma_2 - \frac{n}{n+s}\Sigma_1 = \frac{ns}{(n+s)^2}(x - \mu_1){}^t (x - \mu_1).
    \]
\end{envL}
\begin{proof}
It is easily shown by calculation.
\end{proof}
The rank of $X$ is one.
$X$ is semi-positive definite when $s=1$ and semi-negative definite when $s=-1$.

\begin{envL}[Positive-definiteness of $T_{\alpha}$]\label{lem:private-pos}
If the following two inequalities hold, $T_{\alpha}$ is positive-definite:
    \begin{equation}\label{eqn:positive-cond}
     \frac{n-1}{n} < \tau, \ \     \alpha < \min \left\{ n+1, \ \frac{(n-1)^2 }{\tau n - (n-1)} \right\}.
    \end{equation}
\end{envL}
\begin{proof}
Since $T_{\alpha} = \Sigma_1 \Sigma_{\alpha} \Sigma_2 = \Sigma_2 \Sigma_{\alpha} \Sigma_1$, by Lemma~\ref{prop:matrices}, the positive-definiteness of $T_{\alpha}$ is reduced to the positive-definiteness of $\Sigma_{\alpha}$.
By Lemma~\ref{lem:private-diff}, we have
\[
\Sigma_{\alpha} = (1 - \alpha)\Sigma_1 + \alpha \left( \frac{n}{n+s}\Sigma_1 + X \right) = \left(1 - \frac{s\alpha}{n+s} \right) \Sigma_1 + \alpha X.
\]
When $s=1$, since $\Sigma_1$ is positive-definite and $X$ is semi-positive definite, it is enough to be $\alpha < n+1$.
We consider the case when $s=-1$.
For an arbitrary vector $z \in \R^d$ whose norm is one, we seek a condition where the minimum of ${}^t z \Sigma_{\alpha} z$ is positive.
Here we can consider that the vector $x - \mu_1$ is contained in a ball with a radius $2\sqrt{d}$.
Thus, we obtain the minimum when the following two conditions hold:
\begin{itemize}
    \item $z$ is parallel to the eigenvector of the minimum eigenvalue $\sigma_{min}$ of $\Sigma_1$;
    \item $x-\mu_1$ is parallel to $z$.
\end{itemize}
Hence we see that $\Sigma_{\alpha}$ is positive-definite if 
\begin{eqnarray*}
{}^t z \Sigma_{\alpha} z &=& \left( 1 + \frac{\alpha}{n-1}\right)\sigma_{min} - \alpha \frac{n}{(n-1)^2}4d \\
    &=& \sigma_{min} - \alpha \cdot \frac{4dn - (n-1)\sigma_{min}}{(n-1)^2} \\
    &\ge& \sigma - \alpha \cdot \frac{4dn - (n-1)\sigma}{(n-1)^2} > 0.
\end{eqnarray*}
When the inequalities in Equation~(\ref{eqn:positive-cond}) hold, this inequality also holds.
\end{proof}

\begin{envL}[Upper bound of $L_1$]\label{lem:private-l1-max}
    If $s=1$, then we have
    \[
    L_1 \le \frac{\tau}{(n+1)(n+1-\alpha)},
    \]  
    and if $s=-1$, then we have
    \[
    L_1 \le \frac{\tau}{(n-1)(n-1+\alpha) -\alpha n \tau}.
    \]
\end{envL}
\begin{proof}
Now $\mu_d$ is contained in a ball with a radius $\frac{2\sqrt{d}}{n+s}$ by Lemma~\ref{lem:private-diff} and $\Sigma_{\alpha}$ is positive-definite by Lemma~\ref{lem:private-pos}.
By multiplying the reciprocal of the minimum of ${}^t z \Sigma_{\alpha} z$ for a unit vector $z \in \R^d$ by $\frac{4d}{(n+s)^2}$, we can obtain the maximum of ${}^t \mu_d \Sigma_{\alpha}^{-1} \mu_d$.
Here, we see 
\[
{}^t z \Sigma_{\alpha} z = {}^t z \left(1 - \frac{s \alpha}{n+s}\right) \Sigma_1 z + \frac{s \alpha n}{(n+s)^2} ( {}^t z (x-\mu_1))^2.
\]
Hence when $s=1$, the minimum is 
\[
\left(1 - \frac{\alpha}{n+1}\right) \sigma_{min}.
\]
When $s=-1$, since $x - \mu_1$ is contained in a ball with a radius $2\sqrt{d}$, the minimum is 
\[
\left(1+ \frac{\alpha}{n-1}\right)\sigma_{min} - \frac{\alpha n}{(n-1)^2} \cdot 4d.
\]
Thus, we obtain the inequality.
\end{proof}

\begin{envL}[Lower bound of $L_2$]\label{lem:private-l2-min}
It holds 
\[
L_2 \ge \frac{(1 - \frac{s\alpha}{n+s})^d}{(\frac{n}{n+s})^{\alpha d}} \cdot \min \left\{ 1, \frac{1 + \frac{\alpha ns \tau}{(n + s - s \alpha)(n+s)} }{(1 + \frac{s \tau}{n+s})^{\alpha}} \right\}.
\]
\end{envL}
\begin{proof}
We see that
\[
L_2 := \frac{|(1 - \frac{s\alpha}{n+s}) \Sigma_1 + \alpha X|}{|{\Sigma_1}|^{1-\alpha}|{\frac{n}{n+s} \Sigma_1 + X}|^{\alpha}} = \frac{(1 - \frac{s\alpha}{n+s})^d |I + \frac{n+s}{n + s - s \alpha} \alpha \Sigma_1^{-1} X|}{(\frac{n}{n+s})^{\alpha d}|{ I + \frac{n+s}{n} \Sigma_1^{-1}X}|^{\alpha}}.
\]
Since the rank of $X$ is one and $\Sigma_1^{-1}$ is invertible, the rank of $\Sigma_1^{-1}X$ is also one.
Thus, there is only one non-zero eigenvalue, and it is set as $\lambda$.
We also set $A := (1 - \frac{s\alpha}{n+s})^d/(\frac{n}{n+s})^{\alpha d}$.
Since the other eigenvalues are all zero, we see 
\[
L_2 = \frac{1 + \frac{n+s}{n + s - s \alpha} \alpha \lambda}{(1 + \frac{n+s}{n} \lambda)^{\alpha}} \cdot A.
\]
By differentiating this equation with respect to $\lambda$, we obtain 
\[
\frac{\p L_2}{\p \lambda} = \alpha(\alpha -1)\frac{n+s}{n(n+s - s \alpha)}\cdot\frac{s - (n+s)\lambda}{(1 + \frac{n+s}{n}\lambda)^{\alpha+1}} \cdot A.
\]
We see that $\frac{\p L_2}{\p \lambda} > 0$ when $\frac{s}{n+s} < \lambda$ and $\frac{\p L_2}{\p \lambda} < 0$ when $\frac{s}{n+s} > \lambda$.
Hence the minimum of $L_2$ is obtained at the edges of the range of $\lambda$.

Next, we will find the range of $\lambda$, which is the only one non-zero eigenvalue of $\Sigma_1^{-1} X$.
Since $\Sigma_1$ is positive-definite, we can obtain the spectral decomposition of $\Sigma_1$:
\[
\Sigma_1 = \sum_{i=1}^d \sigma_i p_i {}^t p_i, 
\]
where $\sigma_1, \ldots , \sigma_d$ are the eigenvalues of $\Sigma_1$ and $p_1, \ldots, p_d$ are their eigenvectors whose norms are one.
Since $p_1, \ldots, p_d$ is a basis of $\R^d$, there exist $r_1, \ldots , r_d \in \R$ such that 
\[
x - \mu_1 = \sum_{i=1}^d r_i p_i.
\]
Squaring both sides, we obtain a condition $4d \ge \sum_{i=1}^d r_i^2 > 0$.
Set $e_1 := \sum_{i=1}^d \frac{r_i}{\sigma_i} p_i$.
Then we have
\begin{eqnarray*}
\Sigma_1^{-1} X e_1 &=& \Sigma_1^{-1} \frac{ns}{(n+s)^2} \sum_{i=1}^d r_i p_i ((x - \mu_1) \cdot e_1)\\
&=& \frac{ns}{(n+s)^2} ((x - \mu_1) \cdot e_1) e_1\\
&=& \frac{ns}{(n+s)^2} (\sum_{i=1}^d \frac{r_i^2}{\sigma_i}) e_1.
\end{eqnarray*}
Thus, we have $\lambda = \frac{ns}{(n+s)^2} \sum_{i=1}^d \frac{r_i^2}{\sigma_i}$.
Therefore, we have $0 < \lambda \le \frac{4dn}{(n+1)^2 \sigma_{min}} \le \frac{4dn}{(n+1)^2 \sigma}$ when $s=1$, and $-\frac{4dn}{(n-1)^2 \sigma}  \le -\frac{4dn}{(n-1)^2 \sigma_{min}} \le \lambda <0$ when $s=-1$.
\end{proof}

\section{Numerical Evaluations}

In Theorem~\ref{thm:main-private} and Corollary~\ref{cor:main-public}, we obtain the concrete upper bounds.
Thus, in this section, we compute the value $\ep$ concretely and observe the results.

\subsection{Setting of Numerical Parameters}
We set $d=6$, $\sigma =0.01$ since the number of numerical attributions in Adult Dataset~\cite{Dua:2019} is six and the minimum eigenvalue for the data normalized into $[-1,1]$ is $\sigma_{min}=0.01$.

\subsection{Relation between $\alpha$ and $\ep$}

\begin{figure}[t]
\begin{center}
\includegraphics[width=0.8\linewidth]{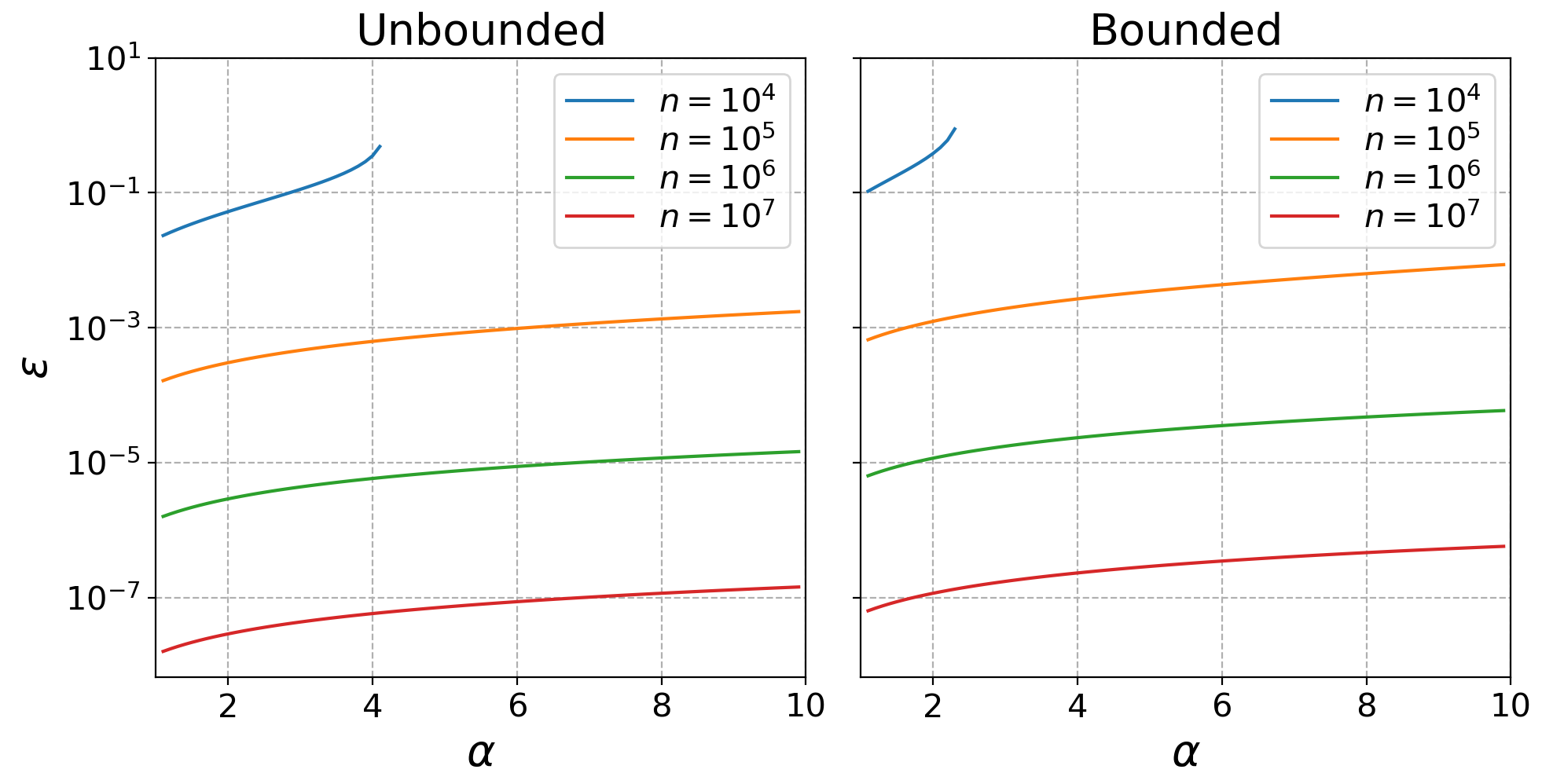}
\caption{$\alpha$-$\ep$ curve ($d=6$,  $\sigma=0.01$) : Vertical axis is logarithmic scale. The curves are drawn for each of the four sample sizes $n$.}
\label{pic:alpha-ep-curve}
\end{center}
\end{figure}

The relations between $\alpha$ and $\ep$ are shown in Fig.~\ref{pic:alpha-ep-curve}  ($\alpha$-$\ep$ curves).
For all curves, $\ep$ is monotonically increasing with respect to $\alpha$.
We also see that as $n$ increases exponentially, $\ep$ becomes smaller at equal intervals on a logarithmic scale.
In particular, if $n=10^4$, the condition in Equation~(\ref{eqn:thm-cond}) is 
\[
\alpha <  c:=  \min \left\{ n+1, \frac{n^2}{\tau(n+1)-n} \right\} \fallingdotseq  4.1679
\]
and the condition in Equation~(\ref{eqn:cor-cond}) is 
\[
\alpha < \frac{c^2}{2c-1} \fallingdotseq 2.3680 .
\]
Thus, the curves stop at these values.

\subsection{The Case Input and Output are the Same Sizes}

For $\alpha=4$, the values of $\ep$ for which the mechanism $\mathcal{M}_G^n$ satisfies $(\alpha, \ep)$-RDP are shown in Table~\ref{tab:n-table}.
By the composition theorem in Proposition~\ref{prop:composition}, the values of $\ep$ are ones in Theorem~\ref{thm:main-private} and Corollary~\ref{cor:main-public} multiplied by $n$.
We can show that values of $\ep$ are within a practical range when $n \ge 10^6$ under both conditions.
In particular, under the unbounded condition, $\ep = 0.5764$ when $n=10^7$, which is very small.
We also see that $\ep$'s under the unbounded condition are four times larger than those under the bounded condition.

\begin{table}[t]
\caption{Values of $\ep$ in the case that input and output are the same size $n$. ($\alpha=4, d=6, \sigma = 0.01$)}\label{tab:n-table}
\hbox to\hsize{\hfil
\begin{tabular}{l|l|l|l|l}\hline
\ $n$ &\  $10^4$ & \ $10^5$ & \ $10^6$ & \ $10^7$ \\ \hline \hline
\ Unbounded $\ep$ \  & \ 3535.17 \ & \ 62.5859 \ & \ {\bf 5.8064} \ & \ {\bf 0.5764} \ \\ \hline
\ Bounded $\ep$\  & \ - \ & \ 266.7349 \ & \ 23.3577 \ & \ {\bf 2.3071} \ \\ \hline
\end{tabular}
\hfil}
\end{table}

\subsection{Translation into $(\ep, \delta)$-DP}

\begin{table}[t]
\caption{Values of $\ep$ in $(\ep, \delta)$-DP with under {\bf unbounded condition} and {\bf bounded condition} ($d=6, \sigma = 0.01$)}\label{tab:epdelta-table}
\hbox to\hsize{\hfil
\begin{tabular}{l|l|l|l|l|l||l|l|l|l|l}
\multicolumn{6}{l}{Case : $n=10^6$} \\ \hline
& \multicolumn{5}{c||}{Unbounded condition} & \multicolumn{5}{c}{Bounded condition} \\ \hline
$\delta$ & $10^{-2}$ & $10^{-5}$ & $10^{-10}$ & $10^{-15}$ & $10^{-20}$ & $10^{-2}$ & $10^{-5}$ & $10^{-10}$ & $10^{-15}$ & $10^{-20}$ \\ \hline \hline
$\alpha = 2$ & 7.499 & 14.407 & 25.920 & 37.433 & 48.946 & {\bf 16.209} & {\bf 23.116} & 34.629 & 46.142 & 57.655 \\
$\alpha = 4$ & {\bf 7.341} & {\bf 9.644} & {\bf 13.482} & 17.319 & 21.157 & 24.893 & 27.195 & {\bf 31.033} & {\bf 34.871} & {\bf 38.708} \\
$\alpha = 7$ & 10.978 & 12.130 & 14.048 & {\bf 15.967} & {\bf 17.886} & 42.046 & 43.198 & 45.116 & 47.035 & 48.954 \\
$\alpha = 10$ & 15.170 & 15.937 & 17.217 & 18.496 & 19.775 & 60.070 & 60.838 & 62.117 & 63.396 & 64.675 \\
$\alpha = 20$ & 30.046 & 30.410 & 31.016 & 31.622 & 32.228 & 123.482 & 123.846 & 124.452 & 125.058 & 125.663 \\
$\alpha = 30$ & 45.624 & 45.863 & 46.260 & 46.657 & 47.054 & 191.710 & 191.948 & 192.345 & 192.742 & 193.139 \\
\hline
\multicolumn{6}{l}{} \\
\multicolumn{6}{l}{Case : $n=10^7$} \\ \hline
& \multicolumn{5}{c||}{Unbounded condition} & \multicolumn{5}{c}{Bounded condition} \\ \hline
$\delta$ & $10^{-2}$ & $10^{-5}$ & $10^{-10}$ & $10^{-15}$ & $10^{-20}$ & $10^{-2}$ & $10^{-5}$ & $10^{-10}$ & $10^{-15}$ & $10^{-20}$ \\ \hline \hline
$\alpha = 2$ & 4.893 & 11.801 & 23.314 & 34.827 & 46.340 & 5.758 & 12.666 & 24.179 & 35.692 & 47.205 \\
$\alpha = 4$ & 2.112 & 4.414 & 8.252 & 12.089 & 15.927 & {\bf 3.842} & 6.145 & 9.982 & 13.820 & 17.658 \\
$\alpha = 7$ & {\bf 1.777} & 2.928 & 4.847 & 6.766 & 8.685 & 4.809 & {\bf 5.960} & {\bf 7.879} & 9.798 & 11.717 \\
$\alpha = 10$ & 1.954 & {\bf 2.722} & {\bf 4.001} & 5.280 & 6.559 & 6.291 & 7.058 & 8.337 & {\bf 9.617} & {\bf 10.896} \\
$\alpha = 20$ & 3.132 & 3.496 & 4.102 & {\bf 4.708} & {\bf 5.313} & 11.838 & 12.201 & 12.807 & 13.413 & 14.019 \\
$\alpha = 30$ & 4.500 & 4.739 & 5.136 & 5.533 & 5.930 & 17.608 & 17.846 & 18.243 & 18.640 & 19.037 \\
\hline
\end{tabular}
\hfil}
\end{table}

By Proposition~\ref{prop:rdp-dp}, we see that $(\alpha, \ep)$-RDP can be translated into $(\ep, \delta)$-DP.

The values translated into $(\ep, \delta)$-DP under the unbounded condition are shown in Table~\ref{tab:epdelta-table}.
When $\delta=0.01$, we see that $\ep=7.341$ for $n=10^6$ and $\ep=1.777$ for $n=10^7$.
When $\delta=10^{-10}$, we also see that $\ep=13.482$ for $n=10^6$ and $\ep=4.001$ for $n=10^7$.
These values are reasonable~\cite{apple,uscensus}.

The results under the bounded condition are shown in Table~\ref{tab:epdelta-table}.
When $\delta=0.01$, we see that $\ep=16.209$ for $n=10^6$ and $\ep=3.842$ for $n=10^7$.
When $\delta=10^{-10}$, we also see that $\ep=31.033$ for $n=10^6$ and $\ep=7.879$ for $n=10^7$.

The values of $\ep$ under the bounded condition are about twice as large as those under the unbounded condition.

\subsection{Summary of Results}

To sum up the results of numerical evaluations, we see the following:
\begin{itemize}
    \item We see that $\ep$ is monotonically increasing with respect to $\alpha$. This result is intuitive.
    \item If $n$ increases exponentially, the curve becomes smaller at equal intervals on a logarithmic scale.
    \item When $n=10^4$, a range where $\alpha$ satisfies the assumption of being very narrow. When $n=10^7$, the value of $\ep$ is practical.
\end{itemize}








\section{Related Work}

In this section, we describe the related work and mention the difference from our result.

\subsection{Differentially Private Synthetic Data Generation}

In synthetic data generation, the post-processing property of differential privacy guarantees that synthetic data generated from differentially private generative parameters also satisfy differential privacy as shown in Fig.~\ref{pic:generative_intro}(b).
Methods to generage differentially private synthetic data for tabular data are classified to two types.

The first type is also called a ``select-measure-generate" scheme~\cite{mckenna2021winning}.
Statistics and (conditional) probability distributions are used as the generative parameters.
Typical statistics are mean vectors and covariance matrices of original datasets.
In particular, synthetic data generation with copulas has been researched actively~\cite{sklar1959fonctions,li2014dpsynthesizer,asghar2020differentially,gambs2021growing}.
To learn conditional distributions, graphical models such as Bayesian networks have been applied to synthetic data generation~\cite{zhang2017priv,zhang2021privsyn,mckenna2022aim,mckenna2019graphical}.

In the second type, generative models with deep neural networks are used to generate synthetic data.
The model parameters trained with the original data are regarded as the generative parameters.
By training deep neural networks with differentially private stochastic gradient descent (DP-SGD)~\cite{abadi2016deep}, we obtain differentially private model parameters.
Methods based on generative adversarial networks (GAN) such as CTGAN~\cite{xu2019modeling}, DPCTGAN~\cite{fang2022dp}, CTAB-GAN~\cite{zhao2021ctab}, and  CTAB-GAN+~\cite{zhao2022ctabplus}, are widely used.
A method based on diffusion model such as TabDDPM~\cite{kotelnikov2022tabddpm} has also attracted attention recently.

In both types of approaches, generative parameters are computed by various differentially private mechanisms~\cite{abadi2016deep,mcsherry2007mechanism} (Fig.~\ref{pic:generative_intro}(b)).
In contrast, we evaluate differential privacy of randomness in data generation when using non-differentially private generative parameters.

\subsection{Privacy Attacks against Synthetic Data Generation}

Many methods empirically evaluate the privacy protection of synthetic data generations from attack success rates of membership inference attacks~\cite{shokri2017membership} and attribute inference attacks~\cite{fredrikson2014privacy}.
Most of them assume that an adversary has access to the target trained model such as GAN~\cite{chen2020gan,Hayes2019logan,hu2021tablegan} and diffusion models~\cite{carlini2023extracting,hu2023membership,duan2023are,matsumoto2023membership}.

On the other hands, there are several methods where an adversary only has access to output synthetic data.
Stadler et al.~\cite{stadler2022synthetic} discussed membership inference attacks and attribute inference attacks for tabular data in such setting, and Oprisanu et al.~\cite{oprisanu2021utility} applied such attacks to genomic data.
Annamalai et al.~\cite{annamalai2023linear} conducted attribute inference with linear reconstruction in this setting.

Although these studies and ours share a common perspective in that they focus on the privacy protection of generated synthetic data alone, these studies differ from ours in that they experimentally evaluate synthetic data generation from an attack perspective.
In contrast, our perspective is to prove R\'{e}nyi differential privacy theoretically.

\subsection{Differential Privacy of Randomness in Synthetic Data Generation}

To the best of our knowledge, only Lin et al.~\cite{lin2021privacy} have evaluated the privacy protection by the randomness in outputs of synthetic data generations.
They theoretically evaluated probabilistic differential privacy~\cite{meiser2018approximate} of GAN-sampled data.
However, the concretely evaluated bound is hard to compute since it needs a GAN's generalization error.
In addition, they assume that training datasets are far larger than the number of model parameters. 
Thus, their main contribution is to give the theoretical bound, but we cannot compute the bound as a concrete numerical value.

In contrast, although we focus on only a simple synthetic data generation, we give the concretely computable bound.

\section{Conclusion}\label{sec:further}

In this paper, we evaluated the privacy protection due to the randomness of synthetic data generation without adding intentional randomness.
We proved R\'{e}nyi differential privacy of a synthetic data generation with a mean vector and covariance matrix (Theorem~\ref{thm:main-private}, Corollary~\ref{cor:main-public}).
We also conducted numerical evaluations using the Adult dataset as a model case.
Concretely, we demonstrated that the mechanism $\mathcal{M}_G^n$ satisfies $(4, 0.576)$-RDP under the unbounded condition and $(4,2.307)$-RDP under the bounded condition (Table~\ref{tab:n-table}).
If they are translated into $(\ep,\delta)$-DP, $\mathcal{M}_G^n$ satisfies $(\ep,\delta)$-DP for a practical $\ep$ (Table~\ref{tab:epdelta-table}).
In future work, we will apply our evaluation method to more advanced synthetic data generation algorithms.

\bibliographystyle{plain}
\bibliography{main}

\begin{thebibliography}{10}

\bibitem{abadi2016deep}
Martin Abadi, Andy Chu, Ian Goodfellow, H~Brendan McMahan, Ilya Mironov, Kunal
  Talwar, and Li~Zhang.
\newblock Deep learning with differential privacy.
\newblock In {\em Proceedings of the 2016 ACM SIGSAC conference on computer and
  communications security}, pages 308--318, 2016.

\bibitem{aggarwal2005k}
Charu~C Aggarwal.
\newblock On k-anonymity and the curse of dimensionality.
\newblock In {\em VLDB}, volume~5, pages 901--909, 2005.

\bibitem{annamalai2023linear}
Meenatchi Sundaram Muthu~Selva Annamalai, Andrea Gadotti, and Luc Rocher.
\newblock A linear reconstruction approach for attribute inference attacks
  against synthetic data.
\newblock {\em arXiv preprint arXiv:2301.10053}, 2023.

\bibitem{apple}
Apple.
\newblock Apple differential privacy technical overview.
\newblock
  \url{https://www.apple.com/privacy/docs/Differential_Privacy_Overview.pdf}.
\newblock Accessed: 2023-02-22.

\bibitem{asghar2020differentially}
Hassan~Jameel Asghar, Ming Ding, Thierry Rakotoarivelo, Sirine Mrabet, and Dali
  Kaafar.
\newblock Differentially private release of datasets using gaussian copula.
\newblock {\em Journal of Privacy and Confidentiality}, 10(2), 2020.

\bibitem{bond2021deep}
Sam Bond-Taylor, Adam Leach, Yang Long, and Chris~G Willcocks.
\newblock Deep generative modelling: A comparative review of vaes, gans,
  normalizing flows, energy-based and autoregressive models.
\newblock {\em IEEE transactions on pattern analysis and machine intelligence},
  2021.

\bibitem{carlini2023extracting}
Nicholas Carlini, Jamie Hayes, Milad Nasr, Matthew Jagielski, Vikash Sehwag,
  Florian Tram{\`e}r, Borja Balle, Daphne Ippolito, and Eric Wallace.
\newblock Extracting training data from diffusion models.
\newblock {\em arXiv preprint arXiv:2301.13188}, 2023.

\bibitem{chen2020gan}
Dingfan Chen, Ning Yu, Yang Zhang, and Mario Fritz.
\newblock Gan-leaks: A taxonomy of membership inference attacks against
  generative models.
\newblock In {\em Proceedings of the 2020 ACM SIGSAC conference on computer and
  communications security}, pages 343--362, 2020.

\bibitem{Dua:2019}
Dheeru Dua and Casey Graff.
\newblock {UCI} machine learning repository, 2017.

\bibitem{duan2023are}
Jinhao Duan, Fei Kong, Shiqi Wang, Xiaoshuang Shi, and Kaidi Xu.
\newblock Are diffusion models vulnerable to membership inference attacks?
\newblock {\em arXiv preprint arXiv:2302.01316}, 2023.

\bibitem{dwork2006differential}
Cynthia Dwork.
\newblock Differential privacy.
\newblock In {\em International Colloquium on Automata, Languages, and
  Programming}, pages 1--12. Springer, 2006.

\bibitem{dwork2014algorithmic}
Cynthia Dwork, Aaron Roth, et~al.
\newblock The algorithmic foundations of differential privacy.
\newblock {\em Found. Trends Theor. Comput. Sci.}, 9(3-4):211--407, 2014.

\bibitem{fang2022dp}
Mei~Ling Fang, Devendra~Singh Dhami, and Kristian Kersting.
\newblock Dp-ctgan: Differentially private medical data generation using
  ctgans.
\newblock In {\em Artificial Intelligence in Medicine: 20th International
  Conference on Artificial Intelligence in Medicine, AIME 2022, Halifax, NS,
  Canada, June 14--17, 2022, Proceedings}, pages 178--188. Springer, 2022.

\bibitem{fredrikson2014privacy}
Matthew Fredrikson, Eric Lantz, Somesh Jha, Simon Lin, David Page, and Thomas
  Ristenpart.
\newblock Privacy in pharmacogenetics: An end-to-end case study of personalized
  warfarin dosing.
\newblock In {\em 23rd $\{$USENIX$\}$ Security Symposium ($\{$USENIX$\}$
  Security 14)}, pages 17--32, 2014.

\bibitem{gambs2021growing}
S{\'e}bastien Gambs, Fr{\'e}d{\'e}ric Ladouceur, Antoine Laurent, and Alexandre
  Roy-Gaumond.
\newblock Growing synthetic data through differentially-private vine copulas.
\newblock {\em Proceedings on Privacy Enhancing Technologies},
  2021(3):122--141, 2021.

\bibitem{gil2013renyi}
Manuel Gil, Fady Alajaji, and Tamas Linder.
\newblock R{\'e}nyi divergence measures for commonly used univariate continuous
  distributions.
\newblock {\em Information Sciences}, 249:124--131, 2013.

\bibitem{goodfellow2014generative}
Ian Goodfellow, Jean Pouget-Abadie, Mehdi Mirza, Bing Xu, David Warde-Farley,
  Sherjil Ozair, Aaron Courville, and Yoshua Bengio.
\newblock Generative adversarial nets.
\newblock {\em Advances in neural information processing systems}, 27, 2014.

\bibitem{harville1998matrix}
David~A Harville.
\newblock Matrix algebra from a statistician's perspective, 1998.

\bibitem{Hayes2019logan}
Jamie Hayes, Luca Melis, George Danezis, and Emiliano~De Cristofaro.
\newblock {LOGAN:} membership inference attacks against generative models.
\newblock {\em Proceedings of Privacy Enhancing Technologies},
  2019(1):133--152, 2019.

\bibitem{hu2021tablegan}
Aoting Hu, Renjie Xie, Zhigang Lu, Aiqun Hu, and Minhui Xue.
\newblock Tablegan-mca: Evaluating membership collisions of gan-synthesized
  tabular data releasing.
\newblock In {\em Proceedings of the 2021 ACM SIGSAC Conference on Computer and
  Communications Security}, pages 2096--2112, 2021.

\bibitem{hu2023membership}
Hailong Hu and Jun Pang.
\newblock Membership inference of diffusion models.
\newblock {\em arXiv preprint arXiv:2301.09956}, 2023.

\bibitem{kairouz2015composition}
Peter Kairouz, Sewoong Oh, and Pramod Viswanath.
\newblock The composition theorem for differential privacy.
\newblock In {\em International conference on machine learning}, pages
  1376--1385. PMLR, 2015.

\bibitem{kifer2011no}
Daniel Kifer and Ashwin Machanavajjhala.
\newblock No free lunch in data privacy.
\newblock In {\em Proceedings of the 2011 ACM SIGMOD International Conference
  on Management of data}, pages 193--204, 2011.

\bibitem{DBLP:journals/corr/KingmaW13}
Diederik~P. Kingma and Max Welling.
\newblock Auto-encoding variational bayes.
\newblock In Yoshua Bengio and Yann LeCun, editors, {\em 2nd International
  Conference on Learning Representations, {ICLR} 2014, Banff, AB, Canada, April
  14-16, 2014, Conference Track Proceedings}, 2014.

\bibitem{kotelnikov2022tabddpm}
Akim Kotelnikov, Dmitry Baranchuk, Ivan Rubachev, and Artem Babenko.
\newblock Tabddpm: Modelling tabular data with diffusion models.
\newblock {\em arXiv preprint arXiv:2209.15421}, 2022.

\bibitem{li2014dpsynthesizer}
Haoran Li, Li~Xiong, Lifan Zhang, and Xiaoqian Jiang.
\newblock Dpsynthesizer: Differentially private data synthesizer for privacy
  preserving data sharing.
\newblock In {\em Proceedings of the VLDB Endowment International Conference on
  Very Large Data Bases}, volume~7, page 1677. NIH Public Access, 2014.

\bibitem{lin2021privacy}
Zinan Lin, Vyas Sekar, and Giulia Fanti.
\newblock On the privacy properties of gan-generated samples.
\newblock In {\em International Conference on Artificial Intelligence and
  Statistics}, pages 1522--1530. PMLR, 2021.

\bibitem{matsumoto2023membership}
Tomoya Matsumoto, Takayuki Miura, and Naoto Yanai.
\newblock Membership inference attacks against diffusion models.
\newblock {\em arXiv preprint arXiv:2302.03262}, 2023.

\bibitem{mckenna2021winning}
Ryan McKenna, Gerome Miklau, and Daniel Sheldon.
\newblock Winning the nist contest: A scalable and general approach to
  differentially private synthetic data.
\newblock {\em arXiv preprint arXiv:2108.04978}, 2021.

\bibitem{mckenna2022aim}
Ryan McKenna, Brett Mullins, Daniel Sheldon, and Gerome Miklau.
\newblock Aim: An adaptive and iterative mechanism for differentially private
  synthetic data.
\newblock {\em arXiv preprint arXiv:2201.12677}, 2022.

\bibitem{mckenna2019graphical}
Ryan McKenna, Daniel Sheldon, and Gerome Miklau.
\newblock Graphical-model based estimation and inference for differential
  privacy.
\newblock In {\em International Conference on Machine Learning}, pages
  4435--4444. PMLR, 2019.

\bibitem{mcsherry2007mechanism}
Frank McSherry and Kunal Talwar.
\newblock Mechanism design via differential privacy.
\newblock In {\em 48th Annual IEEE Symposium on Foundations of Computer Science
  (FOCS'07)}, pages 94--103. IEEE, 2007.

\bibitem{meiser2018approximate}
Sebastian Meiser.
\newblock Approximate and probabilistic differential privacy definitions.
\newblock {\em Cryptology ePrint Archive}, 2018.

\bibitem{mironov2017renyi}
Ilya Mironov.
\newblock R{\'e}nyi differential privacy.
\newblock In {\em 2017 IEEE 30th Computer Security Foundations Symposium
  (CSF)}, pages 263--275. IEEE, 2017.

\bibitem{oprisanu2021utility}
Bristena Oprisanu, Georgi Ganev, and Emiliano De~Cristofaro.
\newblock On utility and privacy in synthetic genomic data.
\newblock {\em arXiv preprint arXiv:2102.03314}, 2021.

\bibitem{rezende2015variational}
Danilo Rezende and Shakir Mohamed.
\newblock Variational inference with normalizing flows.
\newblock In {\em International conference on machine learning}, pages
  1530--1538. PMLR, 2015.

\bibitem{shokri2017membership}
Reza Shokri, Marco Stronati, Congzheng Song, and Vitaly Shmatikov.
\newblock Membership inference attacks against machine learning models.
\newblock In {\em 2017 IEEE Symposium on Security and Privacy (SP)}, pages
  3--18. IEEE, 2017.

\bibitem{sklar1959fonctions}
M~Sklar.
\newblock Fonctions de repartition an dimensions et leurs marges.
\newblock {\em Publ. inst. statist. univ. Paris}, 8:229--231, 1959.

\bibitem{stadler2022synthetic}
Theresa Stadler, Bristena Oprisanu, and Carmela Troncoso.
\newblock Synthetic data--anonymisation groundhog day.
\newblock In {\em 31st USENIX Security Symposium (USENIX Security 22)}, pages
  1451--1468, 2022.

\bibitem{sweeney2002k}
Latanya Sweeney.
\newblock k-anonymity: A model for protecting privacy.
\newblock {\em International Journal of Uncertainty, Fuzziness and
  Knowledge-Based Systems}, 10(05):557--570, 2002.

\bibitem{tao2021benchmarking}
Yuchao Tao, Ryan McKenna, Michael Hay, Ashwin Machanavajjhala, and Gerome
  Miklau.
\newblock Benchmarking differentially private synthetic data generation
  algorithms.
\newblock {\em arXiv preprint arXiv:2112.09238}, 2021.

\bibitem{uscensus}
{United States Census Bureau}.
\newblock Census bureau sets key parameters to protect privacy in 2020 census
  results.
\newblock
  \url{https://www.census.gov/newsroom/press-releases/2021/2020-census-key-parameters.html}.
\newblock Accessed: 2023-02-22.

\bibitem{warner1965randomized}
Stanley~L Warner.
\newblock Randomized response: A survey technique for eliminating evasive
  answer bias.
\newblock {\em Journal of the American Statistical Association},
  60(309):63--69, 1965.

\bibitem{xu2019modeling}
Lei Xu, Maria Skoularidou, Alfredo Cuesta-Infante, and Kalyan Veeramachaneni.
\newblock Modeling tabular data using conditional gan.
\newblock In {\em Advances in Neural Information Processing Systems}, 2019.

\bibitem{zhang2017priv}
Jun Zhang, Graham Cormode, Cecilia~M. Procopiuc, Divesh Srivastava, and Xiaokui
  Xiao.
\newblock Privbayes: Private data release via bayesian networks.
\newblock {\em ACM Trans. Database Syst.}, 42(4), October 2017.

\bibitem{zhang2021privsyn}
Zhikun Zhang, Tianhao Wang, Jean Honorio, Ninghui Li, Michael Backes, Shibo He,
  Jiming Chen, and Yang Zhang.
\newblock Privsyn: Differentially private data synthesis.
\newblock 2021.

\bibitem{zhao2021ctab}
Zilong Zhao, Aditya Kunar, Robert Birke, and Lydia~Y Chen.
\newblock Ctab-gan: Effective table data synthesizing.
\newblock In {\em Asian Conference on Machine Learning}, pages 97--112. PMLR,
  2021.

\bibitem{zhao2022ctabplus}
Zilong Zhao, Aditya Kunar, Robert Birke, and Lydia~Y Chen.
\newblock Ctab-gan+: Enhancing tabular data synthesis.
\newblock {\em arXiv preprint arXiv:2204.00401}, 2022.

\end{thebibliography}

\end{document}